\documentclass[12pt]{article}
\usepackage{amsmath}
\usepackage{amssymb}
\usepackage{amsthm}
\usepackage{graphicx}
\usepackage{natbib}
\usepackage{booktabs}
\usepackage{url}
\usepackage{tikz}
\newcommand{\Tv}{\sigma^{2}_{y}}
\newcommand{\Rv}{\sigma^{2}_{y|x}}

\graphicspath{{figs/}{./}}

\theoremstyle{plain}
\newtheorem{theorem}{Theorem}[section]
\newtheorem{corollary}{Corollary}[section]
\newtheorem{lemma}{Lemma}[section]
\theoremstyle{definition}
\newtheorem{definition}{Definition}[section]

\newcommand{\blind}{0}

\makeatletter
\renewcommand{\@makecaption}[2]{%
  \vskip\abovecaptionskip
  \begingroup\renewcommand{\baselinestretch}{1}\normalsize\selectfont
  \sbox\@tempboxa{\textbf{#1.} #2}%
  \ifdim \wd\@tempboxa >\hsize \textbf{#1.} #2\par \else
    \global\@minipagefalse \hb@xt@\hsize{\hfil\box\@tempboxa\hfil}\fi
  \endgroup \vskip\belowcaptionskip}
\makeatother

\begin{document}

\def\spacingset#1{\renewcommand{\baselinestretch}%
{#1}\small\normalsize} \spacingset{1}

\if1\blind
{
  \title{\bf Designing for Prediction: Planning Formulas for Expected Out-of-Sample $R^2$ and $r$}
  \author{}
  \date{}
  \maketitle
} \fi

\if0\blind
{
  \title{\bf Designing for Prediction: Planning Formulas for Expected Out-of-Sample $R^2$ and $r$}
  \author{Elizabeth Tipton\thanks{Department of Statistics and Data Science, Northwestern University. }\hspace{.2cm}\\
  }
  \date{October 13, 2026}
  \maketitle
} \fi

\thispagestyle{empty}
\vfill

\newpage
\begin{abstract}
\noindent Design calculations let researchers gauge, before data is collected, whether a planned study is likely to detect an effect. In this paper, we develop such calculations for predictive models using ordinary least squares and ridge regression under leave-one-out cross-validation. To do so, we derive closed-form expressions relating the expected out-of-sample $R^2$ and the expected correlation between observed and predicted outcomes to four design-stage quantities: sample size, number of predictors, an assumed population $R^2$ ($\rho^2_{in}$) and a design efficiency factor $D^*$ (which summarizes how leverage is distributed across observations). Unlike the classical shrinkage formulas, which depend only on sample size and dimension, these are sensitive to how covariates are arranged --- to imbalance and skew. A simulation study assesses the accuracy of the formulas against realized values under a variety of data generating mechanisms. We apply the methods to a published study predicting personality from smartphone behavior and show that the design formulas would have indicated, before any analyses were run, that least squares could not be expected to beat the sample mean while ridge could. An R package, \texttt{predpower}, implements the calculations.
\end{abstract}

\noindent%
\vfill

\newpage
\spacingset{1.7}

\section{Introduction}
\label{sec:intro}

Social scientists are increasingly turning from hypothesis testing to predictive models to answer their research questions. For example, \citet{hofman2021} argued in \textit{Nature} that computational social science needs to treat prediction as a goal in its own right, while \citet{almaatouq2024} called for experiments in the behavioral sciences to be designed with integrative, cumulative prediction in mind. These calls follow from evidence that effects vary across settings and populations \citep{bryan2021,yarkoni2022}, resulting in a shift from an over-reliance on a single average estimate towards building models to understand moderation and variation --- and ultimately, to test how well these models serve in predicting individual outcomes. Yet researchers often find the results of these models to be disappointing. For example, when 160 teams applied machine-learning methods to a high-quality birth cohort study to predict six life outcomes, they found that the best predictions were not only inaccurate but also performed only slightly better than a simple benchmark \citep{salganik2020}.

When predictions fail, it is difficult for researchers to know why: is it the outcome, the predictors, the method or the sample? In hypothesis testing, power analysis provides a tool that mitigates this situation, allowing researchers to understand, in advance, if their proposed study is likely to answer the identified question. While design calculations of this kind exist for prediction, they do not focus explicitly on the questions social scientists face. For example, there are sample-size formulas for clinical prediction models \citep{riley2019a,riley2019b} and for polygenic scores \citep{daetwyler2008,dudbridge2013}, but not for the least squares and penalized regression models commonly used by social scientists. Perhaps the closest design analysis for social science is focused on randomized trials \citep{mamakos2025}, but this paper only focuses on a specific use of least squares, not the more general case.

With this shift towards prediction, the statistics used for judging success have also shifted --- from $p$-values and (in-sample) $R^2$ to out-of-sample $R^2$, the correlation between predicted and observed outcomes, AUC and F1, often from cross-validation. These out-of-sample statistics are familiar in name and form to their in-sample counterparts, but they answer a different question. An adjusted in-sample $R^2$ estimates a population quantity, the share of outcome variance the predictors explain. An out-of-sample $R^2$ measures how much of that share a particular study recovers and so depends on the sample size, the number of predictors and the design as well as on the strength of the relationship. For example, \citet{stachl2020} asked whether personality can be predicted from 1{,}821 smartphone behavioral covariates on 624 participants. They evaluated their models out of sample, reported out-of-sample correlations of $0.20$ to $0.40$ and gauged these against effect-size benchmarks from personality psychology. But such benchmarks describe relationships between variables, not how much of a relationship a study of a given size and design can recover and so cannot tell whether a modest out-of-sample value reflects a weak relationship or a limited design.

In this paper, we provide analogous design calculations for use with least squares and penalized (ridge) regression. We focus on two quantities --- the out-of-sample $R^2$ ($R_{oos}^2$) and the out-of-sample correlation ($r_{oos}$) --- and provide formulas for the expected values of these statistics over the sampling distribution of the outcomes under leave-one-out cross-validation (LOOCV). We focus on ridge because it provides analytic tractability with closed-form results. We show that these expected values depend upon parameters that are familiar: the sample size $n$, the number of predictors $p$ and an assumed effect size --- here, the population $R^2$ ($\rho^2_{in}$). The fourth is new: we show that out-of-sample accuracy also depends on how the covariates are arranged across observations --- on collinearity, on imbalance and on leverage. We summarize that arrangement in a single design efficiency factor $D^* \in (0,1]$, computable from the covariates alone. We show that it can be as consequential as the sample size, and unlike the sample size it is frequently improvable, simply by transforming the predictors. We implement these calculations in an R package, \texttt{predpower}, which computes the expected out-of-sample statistics, either from a covariate matrix or from values assumed at the design stage\if0\blind{}; it is available at \url{https://github.com/marianolake/predpower}\else{}; it is provided with the Supplementary Material\fi.

\subsection{Existing approaches}
\label{sec:existing}

Work on the out-of-sample accuracy of regression spans psychometrics, statistics, clinical prediction and genomics, and approaches differ chiefly in two respects: what they predict and what they assume about the covariates. In this paper, we focus on the out-of-sample error, estimated via leave-one-out cross validation and summarized using the statistics $R^2_{oos}$ and $r_{oos}$. We condition on the design matrix $X$ in hand, at finite $n$ and $p$, with no assumption about its distribution. Together these make the calculations available at the design stage, since an expectation over the outcome given $X$ requires the covariates, or a planned version of them, but no outcome. While there is existing work in these two areas, in combination, none address our case. 

\emph{What is predicted.} The classical shrinkage literature targets one of two quantities. \citet{stein1960}, \citet{lord1950}, \citet{nicholson1960}, \citet{darlington1968} and \citet{browne1975} (reviewed in \citealp{yin2001}) give the expected accuracy of a fitted equation for a new case drawn from the population. \citet{efron1986,efron2004} and \citet{mallows1973} instead keep the design points fixed and draw a fresh outcome at each, so that no case is ever new; this refit-and-replicate target is often called in-sample error, and \citet{vanhouwelingen1990} distinguish it and related targets from the actual error of a fitted model on new cases. A second line of work studies the out-of-sample statistic itself, but as something to be estimated from data: \citet{bates2024} show that what cross-validation estimates is closer to the average error of the procedure than to the error of the specific fitted model, and \citet{hawinkel2024} give an unbiased estimator of the out-of-sample $R^2$ with a standard error. Both require an observed outcome. We target the expectation of the leave-one-out statistic, which does not. The closest precedent is \citet{zevallos2026}, who derive the expectation of the out-of-sample $R^2$ under rolling- and expanding-window re-estimation by the same delta-method route we use, but for time-ordered forecasting rather than cross-sectional prediction.

\emph{What is assumed about the covariates.} In the classical shrinkage literature, the formulas assume the predictors are multivariate normal, while in the refit-and-replicate literature, the predictors are treated as fixed, with cases never held out. Existing sample-size methods built on these results, for clinical prediction models \citep{copas1983,copas1987,riley2019a,riley2019b} and for polygenic scores \citep{daetwyler2008,dudbridge2013}, summarize the design by $n$, $p$ and an assumed signal strength; the clinical methods target the calibration slope of \citet{copas1983}, the factor by which fitted coefficients must be shrunk to predict new cases without systematic over- or under-prediction. A large asymptotic literature lets $n$ and $p$ grow together and characterizes the covariates by their population distribution \citep{dobriban2018,elkaroui2018,rosset2020,atanasov2025,bellec2025}, and \citet{mamakos2025} characterize them by how the moments of the sampled population differ from those of a target population. Importantly, these assumptions are consequential: averaging over a normal population or redrawing outcomes at fixed design points removes any dependence on how the covariates are arranged across observations; the resulting existing formulas thus depend on the design only through $n$ and $p$. In this paper, we show that holding a case out of a fixed design, as we do, also depends on the arrangement of the covariates, and Section~\ref{sec:ols} shows how.

\subsection{Overview of the paper}

Throughout this paper, we assume one outcome per unit, homoscedastic errors and no clustering. In Section~\ref{sec:ols}, we develop the design formulas for ordinary least squares and show how they differ from existing formulas. Section~\ref{sec:ridge} extends them to ridge regression. Section~\ref{sec:simulation} provides a simulation study that assesses the accuracy of these formulas under the required assumptions and away from them and compares them with the lasso. In Section~\ref{sec:illustration}, we illustrate the full framework on a published prediction study. We close by discussing what these formulas imply for applied practice and how they relate to other recent work. 

\section{Design Formulas Under Ordinary Least Squares}
\label{sec:ols}

Consider $i = 1, \dots, n$ independent observations following the ordinary regression model
\[
y_i = x_i'\beta + \epsilon_i, \qquad \epsilon_i \sim (0, \Rv),
\]
where $x_i$ is a row of $p+1$ entries --- an intercept and $p$ covariates --- and $\beta$ is the coefficient vector. Throughout, the $p$ covariates are standardized to mean zero and variance one. Stack the rows into the $n \times (p+1)$ design matrix $X$; write $X=[\mathbf{1}\;\; Z]$, with $Z$ the standardized covariate block and $\beta_z$ its coefficients. The ordinary least-squares estimator is $\hat\beta = (X'X)^{-1}X'y$, with $\mathrm{Var}(\hat\beta) = \Rv(X'X)^{-1}$, and the prediction for unit $i$ when that unit is held out of the fit is $\hat y_{i|-i} = x_i'\hat\beta_{-i}$. $E[R^2_{oos}]$ and $E[r_{oos}]$ below, and every formula in Sections~\ref{sec:ols} and~\ref{sec:ridge}, are properties of these held-out predictions.

In-sample performance is summarized by $\rho^2_{in} = \mathrm{Var}(x_i'\beta)/\mathrm{Var}(x_i'\beta + \epsilon_i) = \beta'\Sigma\beta/\Tv$, the share of population outcome variance carried by the systematic part, with $\Sigma$ the population covariance of $x_i$; its sample analogue, which we call $R^2_{in}$, is the well-known (yet biased) estimator, commonly reported in model statistics. We develop closed-form design formulas for two out-of-sample statistics, the out-of-sample coefficient of determination $R^2_{oos}$ and the out-of-sample correlation $r_{oos}$, defined in Section~\ref{sec:loo}. We show below that the familiar in-sample relation $r_{in}^2 = R^2_{in}$ does \emph{not} survive out of sample.

Write $\Tv := \mathrm{Var}(y_i) = \mathrm{Var}(x_i'\beta) + \Rv$ for the unconditional variance of the outcome, so that $\rho^2_{in} = \mathrm{Var}(x_i'\beta)/\Tv$ and $\Rv = (1-\rho^2_{in})\Tv$; if the outcome is standardized, $\Tv=1$. Note that $\Tv$ is a scale constant rather than a design input: because $\Rv = (1-\rho^2_{in})\Tv$, it cancels from the formulas below, which depend on the outcome only through the ratio $D/\Tv$. Here $p$ excludes the intercept.

\subsection{The leave-one-out statistics}
\label{sec:loo}

Both statistics are computed from the $n$ held-out predictions $\hat y_{i|-i}$. With $\bar{\hat y}$ the mean of these predictions,
\begin{equation}
R^2_{oos} = 1-\frac{\sum_i (y_i-\hat y_{i|-i})^2}{\sum_i (y_i-\bar y)^2},
\qquad
r_{oos} = \frac{\sum_i (y_i-\bar y)(\hat y_{i|-i}-\bar{\hat y})}{\sqrt{\sum_i (y_i-\bar y)^2\sum_i (\hat y_{i|-i}-\bar{\hat y})^2}}.
\label{eq:estimators}
\end{equation}
The first compares the held-out prediction errors with those of predicting the training mean: zero means the model does no better than the mean, and negative values are possible. The second is the correlation between predicted and observed values.

In practice $R^2_{oos}$ and $r_{oos}$ are computed by $k$-fold cross-validation, with leave-one-out cross-validation (LOOCV) the special case $k=n$. We focus on LOOCV for two reasons. First, it is the case classical cross-validation theory already treats exactly \citep{stone1974,allen1974}. Second, and more practically, it is the most tractable case. Removing a single observation from a fitted model has an exact algebraic shortcut, the PRESS identity: with $h_i = x_i'(X'X)^{-1}x_i$ the familiar hat-matrix leverage of observation $i$ and $e_i = y_i - x_i'\hat\beta$ its in-sample residual,
\begin{equation}
y_i - \hat y_{i|-i} \;=\; \frac{e_i}{1-h_i},
\qquad\text{equivalently}\qquad
\hat y_{i|-i} \;=\; y_i - \frac{e_i}{1-h_i},
\label{eq:press}
\end{equation}
exactly, for any design \citep{allen1974}. All $n$ held-out predictions are therefore available from a single fit, and each held-out error depends on the design only through that observation's leverage; this is what allows the moments below to be written in closed form from the design, before any outcome is observed. The same shortcut extends to $k$-fold cross-validation by deleting blocks of observations, but the result then depends on the particular random partition into folds, which leave-one-out avoids.

\begin{definition}[Leave-one-out coefficient of determination and reliability]
\label{def:Sloo}
With $\bar{\hat y}$ the mean of the held-out predictions, write
\[
S_{yy} := \sum_i (y_i-\bar y)^2, \qquad
\mathrm{PRESS} := \sum_i (y_i - \hat y_{i|-i})^2,
\]
\[
S_{\hat y\hat y} := \sum_i (\hat y_{i|-i}-\bar{\hat y})^2, \qquad
S_{y\hat y} := \sum_i (y_i-\bar y)(\hat y_{i|-i}-\bar{\hat y}).
\]
Here $\mathrm{PRESS}$ is the prediction sum of squares of \citet{allen1974}, which~\eqref{eq:press} computes from a single fit.
The expectations of the two leave-one-out statistics are
\[
E\big[R^2_{oos}\big] \;=\; 1-E\!\left[\frac{\mathrm{PRESS}}{S_{yy}}\right], \qquad
E\big[r_{oos}\big]\;=\;E\!\left[\frac{S_{y\hat y}}{\sqrt{S_{yy}S_{\hat y\hat y}}}\right],
\]
expectations taken over the outcome with $X$ held fixed; Supplement~A.13 gives the corresponding results when the rows of $X$ are themselves sampled. Throughout, $R^2_{oos}$ and $r_{oos}$ denote the statistics of~\eqref{eq:estimators}, computed from data, and the formulas below give their expectations, written $E[R^2_{oos}]$ and $E[r_{oos}]$ --- design-stage quantities that can be calculated before any outcome is observed.
\end{definition}

\subsection{The role of study design}
\label{sec:design}

The denominator $S_{yy}$ of $R^2_{oos}$ does not involve the held-out predictions at all, so whatever the design contributes enters through $\mathrm{PRESS}$, the sum of squared held-out errors. Taking its expectation one observation at a time, each term decomposes exactly into two pieces:
\[
E[\mathrm{PRESS}] = \Rv + x_i'\,\mathrm{Var}(\hat\beta_{-i})\,x_i,
\]
because the held-out sampling error is independent of $\hat\beta_{-i}$. The first piece is irreducible; the second is affected by the design of the study. By the Sherman--Morrison identity \citep{sherman1950} --- the same algebra that gives the PRESS identity~\eqref{eq:press} --- removing one row from $X'X$ has a closed form:
\[
x_i'(X'X - x_ix_i')^{-1}x_i = \frac{h_i}{1-h_i},
\]
with $h_i$ the leverage of Section~\ref{sec:loo}. Averaging this over the $n$ observations gives the design prediction penalty.

\begin{definition}[Design prediction penalty]
\label{def:D}
Let $D$ be the average, over the $n$ observations, of how much estimation error in $\hat\beta$ costs a held-out prediction, where
\[
D = \frac{\Rv}{n}\sum_i \frac{h_i}{1-h_i}.
\]
\end{definition}

\noindent Summing the decomposition above over observations then gives $E[\mathrm{PRESS}] = n(\Rv + D)$.

A basic identity of OLS is that the leverages sum to the number of columns, $\sum_i h_i = p+1$, which is exact for any design however collinear or imbalanced. If leverage were spread perfectly evenly --- $h_i = (p+1)/n$ for every $i$ --- the penalty would take the closed form
\[
D_{even} = \Rv\cdot\frac{p+1}{n-p-1}.
\]
Because $h/(1-h)$ is strictly convex in leverage, Jensen's inequality guarantees $D \ge D_{even}$ whenever leverage is unevenly spread, with equality only in the perfectly balanced case. Both penalties have familiar sample analogues: PRESS estimates $n(\Rv+D)$, and generalized cross-validation \citep{craven1979,golub1979} replaces each leverage with the average, in effect estimating with $D_{even}$.

\begin{definition}[Design efficiency factor]
\label{def:Dstar} Let $D^*$ be the design efficiency factor, where
\[
D^* = D_{even}/D \in (0,1].
\]
$D^*$ depends only on the leverages $h_i$, on $p+1$ and on $n$. It equals $1$ exactly when leverage is balanced and falls below $1$ by exactly what \emph{uneven leverage} costs relative to that ceiling.
\end{definition}

What drives $D^*$ is whether a few observations sit far from the rest. In the Supplement we provide an exploration across different design matrix patterns, holding $n = 100$ and $p = 10$ (Supplement~B, Table~B4). Imbalance across categories costs little while every category is reasonably large --- an imbalanced $90/5/5$ category costs a little ($0.949$) --- but the cost grows quickly as the smallest category shrinks. For example, a $96/2/2$ split shifts $D^*$ even further ($0.843$), and a category with a single member makes leave-one-out prediction undefined. Skewed covariates produce the same configuration more reliably, costing a quarter of the design's efficiency on average ($0.749$) and pushing maximum leverage from $0.11$ to $0.68$. Collinearity with Gaussian covariates results in $D^*$ at $0.979$, but combined with skewed covariates it falls to $0.549$. Skew is the mechanism behind Section~\ref{sec:illustration}, where $74\%$ of candidate predictors have absolute skewness above $2$ on their observed values and realized $D^*$ falls to $0.18$.

\subsection{Moments of the leave-one-out statistics}
\label{sec:moments}

Our goal is design-oriented formulas for $E[R^2_{oos}]$ and $E[r_{oos}]$, and these follow from the moments of the four sums of Definition~\ref{def:Sloo}. The intercept drops out of everything below, because $X'X$ is block diagonal when the covariates are centered.

\begin{lemma}[Moments of the held-out cross-products]
\label{lem:moments}
Fix the design $X$ and the coefficient vector $\beta$, and let $\tilde\rho^2_{in} := \beta_z'Z'Z\beta_z/[(n-1)\Tv]$ be the realized in-sample share of that design. Writing $S_{y\hat y}$, $S_{yy}$, $S_{\hat y\hat y}$ and $\mathrm{PRESS}$ for the sums of \eqref{eq:estimators}, $\bar r =\bar{y} - \bar{\hat{y}}$ and taking expectations over the outcome alone,
\[
E[S_{yy}] = (n-1)\Tv,
\qquad
E[S_{y\hat y}] = \Tv\,(n\tilde\rho^2_{in}-1),
\]
\[
E[S_{\hat y\hat y}] = \Tv\,(n\tilde\rho^2_{in}-1) + nD - n\,E[\bar r^{\,2}],
\qquad
E[\mathrm{PRESS}] = n\big(\Rv+D\big).
\]
\end{lemma}

\noindent The expression for $E[\mathrm{PRESS}]$ is the decomposition of Section~\ref{sec:design}, summed over observations; the other three are new. Averaging over the sampling of the rows as well, the same identities hold with $\tilde\rho^2_{in}$ replaced by the population share $\rho^2_{in} = \beta'\Sigma\beta/\Tv$, because the sample covariance is unbiased for $\Sigma$. Since $\rho^2_{in}$ depends on the outcome, it must be assumed whether or not the covariates are already in hand; it is the population share that is assumed throughout. The proof is in Supplement~A.3, together with the exact expression for $E[\bar r^{\,2}]$. 

In this paper, we use these moments to derive approximate formulas for $E[R^2_{oos}]$ and $E[r_{oos}]$. These involve two approximations. First, they rely on first-order delta-method approximations to the ratio estimators: Definition~\ref{def:Sloo} gives $E[R^2_{oos}] = 1-E[\mathrm{PRESS}/S_{yy}]$ exactly, whereas the formulas approximate this with $1-E[\mathrm{PRESS}]/E[S_{yy}]$, the two agreeing to $O(n^{-1})$. Second, in $E[S_{\hat y\hat y}]$ we drop the term $E[\bar r^{\,2}]$, the only term in the lemma that is not a function of $n$, $p$, $\rho^2_{in}$ and the leverage sequence. It has no bearing on $E[R^2_{oos}]$, and for $E[r_{oos}]$ it makes the formula a conservative lower bound. Its relative size is bounded by the largest leverage alone. Across the designs of Section~\ref{sec:simulation} it is at most $0.6\%$ when leverage is balanced; when leverage is imbalanced it shrinks with $n$, reaching $2.1\%$ under least squares and $12\%$ under ridge at $n=30$ but at most $1\%$ at $n=600$ (Supplement~A.12).

In combination, our formulas are thus $E[R^2_{oos}] = F_R + O(n^{-1})$ and $E[r_{oos}] = F_r + O(n^{-1})$, for some closed forms $F_R$ and $F_r$; we write $\approx$ in-line for readability once this is stated, since the source and order of the error is the same throughout.

\subsection{Two design-oriented statistics}

The formulas below take four design-stage inputs: $n$, $p$, an assumed population $R^2$ ($\rho^2_{in}$) and $D^*$. They are not meant for use with the sample $R^2_{in}$, which is biased; if an estimate must be used, the adjusted $R^2$ \citep{wherry1931} should be substituted, with the caveat --- familiar from power analysis --- that this adds uncertainty to the design calculation. We discuss this further in the illustration.

\begin{theorem}[The OLS design formula]
\label{thm:ols}
For any design,
\[
E\big[R^2_{oos}\big] \;\approx\; \frac{n\psi-1}{n-1} ,\qquad
\psi:=\rho^2_{in}-\dfrac{(1-\rho^2_{in})(p+1)}{(n-p-1)D^*}.
\]
\end{theorem}

\noindent\emph{Proof in Supplement~A.4.}

\noindent This differs from the classical, refit-and-replicate coefficient of determination in three ways, all consequences of evaluating on genuinely held-out predictions rather than a fixed design refit with new outcomes. The largest is in the design term. Refit-and-replicate charges $(1-\rho^2_{in})(p+1)/n$. This is Efron's optimism \citep{efron1986,efron2004}, the quantity behind Mallows's $C_p$ \citep{mallows1973}, and it does not depend on how $X$ is arranged because $\sum_i h_i=p+1$ exactly. Leave-one-out instead charges $(1-\rho^2_{in})(p+1)/[(n-p-1)D^*]$, which is larger by a factor of $n/(n-p-1)$ even when leverage is evenly spread and by a further $1/D^*$ when it is not. The other two are small-sample corrections: a numerator of $n\rho^2_{in}-1$ in place of $(n-1)\rho^2_{in}$ and a factor of $n/(n-1)$ on the design term. It is removing observation $i$ from the fit that introduces the $h_i/(1-h_i)$ term and with it the dependence on how leverage is distributed.

\begin{corollary}[Minimum required $R^2$]
\label{cor:minR2}
$E[R^2_{oos}] \ge 0$ exactly when $\rho^2_{in} \ge \rho^2_{in,min}$, where
\[
\rho^2_{in,min} \;=\; \frac{1}{n} \;+\; \left(1-\frac{1}{n}\right)\cdot\frac{p+1}{(p+1)+ (n-p-1) D^*}.
\]
\end{corollary}

\noindent\emph{Proof in Supplement~A.5.}

\noindent This is the predictive analogue of a minimum detectable effect size: the share of outcome variance the predictors must explain in the population for the model to beat the sample mean. 

\begin{theorem}[Predictive reliability]
\label{thm:reliability}
The leave-one-out correlation between observed and predicted outcomes is
\[
E\big[r_{oos}\big] \;\approx\; \mathrm{sgn}\big(n\rho^2_{in}-1\big)\,\sqrt{\,\frac{(n\rho^2_{in}-1)(n-p-1)D^*}{(n\rho^2_{in}-1)(n-p-1)D^* \;+\; n(1-\rho^2_{in})(p+1)}\cdot\frac{n\rho^2_{in}-1}{n-1}\,}.
\]
The right-hand side is a lower bound in absolute value (Supplement~A.12).
\end{theorem}

\noindent\emph{Proof in Supplement~A.6.}

\begin{corollary}[Reliability at the break-even point]
\label{cor:halfcalib-ols}
At the break-even point of Corollary~\ref{cor:minR2}, where $E[R^2_{oos}]=0$, \[
E[r_{oos}]^2 = \tfrac12\cdot(n\rho^2_{in}-1)/(n-1).
\]
\end{corollary}

\noindent\emph{Proof in Supplement~A.6.}

\noindent It is tempting to expect $E[R^2_{oos}]=E[r_{oos}]^2$, mirroring the familiar in-sample identity. To see why this does not hold, write $C := E[S_{y\hat y}]/E[S_{\hat y\hat y}]$ (notation of Definition~\ref{def:Sloo}; like the formulas above, it omits the small term discussed in Supplement~A.12) for the calibration slope of \citet{copas1983}, the slope from regressing the outcome on the held-out prediction and the quantity targeted by the clinical sample-size methods of Section~\ref{sec:existing}. Because those methods hold no case out, their slope carries no leverage term; here, through $D$, it does. Comparing these two expectations,
\[
E[r_{oos}]^2 - E[R^2_{oos}] \;\approx\; \frac{n\rho^2_{in}-1}{n-1}\cdot\frac{(1-C)^2}{C},
\]
which is non-negative whenever $\rho^2_{in}>1/n$ and $C>0$ and zero only when $C=1$. The mechanism is calibration. Correlation is scale-invariant --- $E[r_{oos}]$ asks whether $\hat y_{i|-i}$ moves with $y_i$, not whether it moves by the right amount --- while $E[R^2_{oos}]$ is affected by both. The result is that when the predictions are miscalibrated ($C \ne 1$), $E[r_{oos}]^2$ overstates accuracy \citep[see][for a vivid applied case]{pratiwi2024}. This also explains why the in-sample identity $R^2_{in}=r_{in}^2$ holds: it is not a general fact about $R^2$ and correlation but a restatement of the normal equations, $X'(y-X\hat\beta)=0$, which force fitted values orthogonal to residuals on the training data. The identity fails once that orthogonality is broken, whether by predicting a new case, leaving one out or adding a ridge penalty.

\section{Design Formulas Under Ridge Regression}
\label{sec:ridge}

Uneven leverage has a natural fix: shrink $\hat\beta$ toward zero. This is the idea of penalized regression, especially helpful when $p>n$ but usable more generally. We focus on ridge specifically, rather than the lasso, because ridge's bias and variance separate analytically, which is what makes a closed-form design formula possible. Writing $\beta = (\alpha,\beta_1,\dots,\beta_p)'$ with $\alpha$ the intercept, the ridge estimator solves
\[
\hat\beta(\kappa) = \arg\min_\beta \left\{\sum_i (y_i - x_i'\beta)^2 + \kappa \sum_j \beta_j^2\right\},
\]
with the penalty running over the $p$ covariate coefficients only --- the intercept is not shrunk.

Let $K$ be the $(p+1)\times(p+1)$ matrix that is zero except for a penalty $\kappa$ on the diagonal entries belonging to the standardized, non-intercept covariates: $K = \kappa P$, where $P$ is the diagonal projection matrix with a zero in the intercept's slot and ones elsewhere. Then $\hat\beta(\kappa) = (X'X+K)^{-1}X'y$. Unlike $\hat\beta$, it is biased: $E[\hat\beta(\kappa)] - \beta = -(X'X+K)^{-1}K\beta$.

\begin{definition}[The ridge design prediction penalty]
\label{def:Dridge}
\[
D(\kappa) := \frac{1}{n}\sum_i \left\{ x_i'\,\mathrm{Var}(\hat\beta_{-i}(\kappa))\,x_i \;+\; \Big(x_i'\big[E\hat\beta_{-i}(\kappa)-\beta\big]\Big)^2 \right\},
\]
generalizing Definition~\ref{def:D}, which is the special case $\kappa=0$.
\end{definition}

A difficulty for design stage analysis with ridge is that the results depend upon the penalization parameter $\kappa$. When outcomes are available, this value is typically selected using cross-validation. Without having $y$ on hand, we instead derive an optimal $\kappa$ based upon an assumed structure of the $\beta$s. 

\begin{theorem}[Design-implied ridge penalty]
\label{thm:ridgepenalty}
Assume the $p$ covariate coefficients are a draw from a mean-zero, isotropic distribution (variance $\nu^2$), independent of $X$ and the residual noise. Then the penalty minimizing the expected $D(\kappa)$, for any design, is
\[
\kappa^* = \frac{p(1-\rho^2_{in})}{\rho^2_{in}},
\]
independent of the outcome variation $\Tv$.
\end{theorem}

\noindent\emph{Proof sketch (full derivation in Supplement~A.7).} Treating the $p$ standardized covariate coefficients as a mean-zero, isotropic draw with variance $\nu^2 = \rho^2_{in}\Tv/p$ turns the unknown $\beta$ into a single design parameter. Taking expectations of the per-observation bias and variance terms over this prior and differentiating with respect to $\kappa$ shows that every per-observation term in $D(\kappa)$ is stationary at exactly $\kappa^* = \Rv/\nu^2$, regardless of the design or leverage distribution. Substituting $\nu^2 = \rho^2_{in}\Tv/p$ and $\Rv = \Tv(1-\rho^2_{in})$ gives $\kappa^*$ with $\Tv$ canceling exactly. This is exactly the penalty proposed by \citet{hoerl1975}; what is new here is that $p$ and a hypothesized $\rho^2_{in}$ are design-stage quantities.

For a penalty $\kappa$, write $B(\kappa) := X'X+\kappa P$ and $g_i(\kappa) := x_i'B(\kappa)^{-1}x_i$ for the ridge leverage of observation $i$, the analogue of $h_i$.

\begin{lemma}[Moments of the held-out cross-products under ridge]
\label{lem:ridgemoments}
At penalty $\kappa$, under the isotropic prior of Theorem~\ref{thm:ridgepenalty}, and writing $S_{y\hat y}$, $S_{yy}$, $S_{\hat y\hat y}$ and $\mathrm{PRESS}$ for the sums of \eqref{eq:estimators} and $\bar r$ for the mean held-out residual,
\[
E[S_{yy}] \;=\; (n-1)\Tv, \qquad E[\mathrm{PRESS}] \;=\; n\big(\Rv+D(\kappa)\big),
\]
\[
E[S_{y\hat y}] \;=\; \rho^2_{in}\Tv(n-1) \;-\; \frac{\rho^2_{in}\Tv\kappa}{p}\sum_i\frac{g_i(\kappa)-1/n}{1-g_i(\kappa)} \;-\; \Rv,
\]
\[
E[S_{\hat y\hat y}] \;=\; 2\,E[S_{y\hat y}] \;-\; E[S_{yy}] \;+\; E[\mathrm{PRESS}] \;-\; n\,E[\bar r^{\,2}].
\]
\end{lemma}

\noindent Setting $\kappa=0$ recovers Lemma~\ref{lem:moments}, since the middle term vanishes. The proof, common to both, is in Supplement~A.3.

\subsection{The design parameter $D^*(\kappa^*)$}

\begin{theorem}[Ridge design efficiency at $\kappa^*$]
\label{thm:ridgeeff}
When $\kappa=\kappa^*$, let $g_i(\kappa^*)$ be the ridge leverage at $\kappa^*$. Then the design penalty --- in expectation over the errors and the coefficient prior of Theorem~\ref{thm:ridgepenalty}, conditional on $X$ --- reduces exactly to
\[
D(\kappa^*) = \frac{\Rv}{n} \cdot \sum_i \frac{g_i(\kappa^*)}{1-g_i(\kappa^*)}.
\]
\end{theorem}

\noindent This is structurally identical to Definition~\ref{def:D}'s $D$, with $g_i(\kappa^*)$ replacing $h_i$. The full derivation, including the Sherman--Morrison--Woodbury downdate of $B$ and the cancellation of the bias and variance terms at $\kappa=\kappa^*$, is given in Supplement~A.8.

\begin{definition}[Effective parameters $p_{eff}(\kappa^*)$]
\label{def:peff}
The effective number of parameters when $\kappa = \kappa^*$ is
\[
p_{eff}(\kappa^*) = \sum_j \frac{\lambda_j}{\lambda_j+\kappa^*},
\]
where $\lambda_1,\dots,\lambda_p$ are the eigenvalues of $Z'Z$; the unpenalized intercept contributes the separate $1$ below.
\end{definition}

\begin{definition}[Ridge design efficiency factor]
\label{def:Dstarridge}
At $\kappa=\kappa^*$, \[
D^*(\kappa^*) := D_{even}(\kappa^*)/D(\kappa^*),
\]
where $D_{even}(\kappa^*) := \Rv \cdot [(p_{eff}(\kappa^*)+1)/n] / [1-(p_{eff}(\kappa^*)+1)/n]$, using the exact trace identity $\sum_i g_i(\kappa^*) = 1+p_{eff}(\kappa^*)$ (Supplement~A.2), the ridge analogue of $\sum_i h_i = p+1$. As in Definition~\ref{def:Dstar}, $D^*(\kappa^*) \in (0,1]$ by Jensen's inequality, equal to 1 exactly when ridge leverage is perfectly even.
\end{definition}

\subsection{Two design-oriented statistics under penalization}

We now return to $E[R^2_{oos}]$ and $E[r_{oos}]$ under the penalized model.

\begin{theorem}[The ridge design formula]
\label{thm:ridgeR2}
At $\kappa=\kappa^*$, for any design,
\[
E\big[R^2_{oos}(\kappa^*)\big] \;\approx\; \frac{n\psi(\kappa^*)-1}{n-1} ,\qquad
\psi(\kappa^*):=\rho^2_{in}-\dfrac{(1-\rho^2_{in})\,(p_{eff}(\kappa^*)+1)}{(n-p_{eff}(\kappa^*)-1)\,D^*(\kappa^*)}.
\]
\end{theorem}

\noindent\emph{Proof in Supplement~A.9.}

\noindent This is Theorem~\ref{thm:ols} with $p_{eff}(\kappa^*)$ in place of $p$ and $D^*(\kappa^*)$ in place of $D^*$.

\begin{corollary}[Minimum required signal under ridge]
\label{cor:minridge}
At $\kappa=\kappa^*$, $E[R^2_{oos}(\kappa^*)]\ge 0$ exactly when $\rho^2_{in} \ge \rho^2_{in,min}$, where
\begin{equation}
\rho^2_{in,min} \;=\; \frac{1}{n} \;+\; \left(1-\frac{1}{n}\right)\cdot
  \frac{p_{eff}(\kappa^*)+1}{\big(p_{eff}(\kappa^*)+1\big) + \big(n-p_{eff}(\kappa^*)-1\big)D^*(\kappa^*)}.
\label{eq:ridgemin}
\end{equation}
\end{corollary}

\noindent\emph{Proof in Supplement~A.10.}

As in Corollary~\ref{cor:minR2}, the threshold adds a flat $1/n$ shift to the threshold obtained by setting $\rho^2_{in}-D(\kappa^*)/\Tv$ to zero. Unlike Corollary~\ref{cor:minR2}, however, $\kappa^*$ is itself a function of $\rho^2_{in}$ (Theorem~\ref{thm:ridgepenalty}), so $\rho^2_{in,min}$ appears on both sides through $p_{eff}(\kappa^*)=\sum_j \lambda_j/(\lambda_j+\kappa^*)$. For a general design the threshold must therefore be obtained by bracketed root-finding. The root is unique when it exists, because at $\kappa^*$ the expected out-of-sample $R^2$ is convex in $\rho^2_{in}$; when $p>n$ and $D^*(\kappa^*)$ is low there may be none below one, in which case no value of $\rho^2_{in}$ lets ridge be expected to beat the sample mean (Supplement~A.10).

Under least squares the design enters the formulas only through $D^*$, so a planner needs $n$, $p$, an assumed $\rho^2_{in}$ and a value for $D^*$. Under ridge it enters through two quantities, $D^*(\kappa^*)$ and $p_{eff}(\kappa^*)$, and both depend on the covariate matrix: the first through the ridge leverages, the second through the eigenvalues of $Z'Z$. When the covariates are in hand but the outcome is not --- the situation of Section~\ref{sec:illustration} --- both are computed exactly from $X$ and the assumed $\rho^2_{in}$, and nothing further is needed. When no design yet exists, both must be supplied. $D^*(\kappa^*)$ can be assumed directly, as $D^*$ is under least squares. For $p_{eff}(\kappa^*)$, we approximate the eigenvalues by those of a design whose standardized covariates are uncorrelated. Even then the eigenvalues of $Z'Z$ are not all equal: with $p$ comparable to $n$ they spread out by sampling alone, following the law of \citet{marchenko1967} (see also \citealp{dobriban2018}). Writing $\gamma = p/(n-1)$ and $k = \kappa^*/(n-1)$, that law gives the closed form
\[
p_{eff}(\kappa^*) \;\approx\; p - \frac{n-1}{2}\left[\sqrt{(1-\gamma+k)^2 + 4\gamma k} \;-\; (1-\gamma+k)\right],
\]
a function of $n$, $p$ and the assumed $\rho^2_{in}$ alone (Supplement~A.14). Across the simulation of Section~\ref{sec:simulation} it reproduces the realized $p_{eff}(\kappa^*)$ for uncorrelated covariates at every ratio of $p$ to $n$, including $p>n$. Correlated or heavy-tailed covariates spread the eigenvalues further and lower $p_{eff}(\kappa^*)$, so for them the approximation overstates it. In the simulation the overstatement grows with collinearity, leverage and dimension; it is largest for collinear, high-leverage designs with $p>n$, where the approximation roughly doubles $p_{eff}(\kappa^*)$ (Supplement Table~A1). For a given $D^*(\kappa^*)$, the resulting threshold is therefore conservative. We use it only where no covariate matrix is assumed (Figure~\ref{fig:prospect}); the simulation of Section~\ref{sec:simulation} and the analysis of Section~\ref{sec:ill-posthoc} use each design's own eigenvalues.

As in Section~\ref{sec:ols}, $D(\kappa^*)/\Tv$ is the share of prediction variance that is estimation error rather than signal. Under least squares the results hold for any coefficient vector $\beta$ (Lemma~\ref{lem:moments}); under ridge they do not, because the fit is deliberately biased and its bias depends on $\beta$, so the isotropic prior of Theorem~\ref{thm:ridgepenalty} is needed to average over it. The PRESS shortcut \eqref{eq:press} generalizes exactly, with ridge leverage $g_i(\kappa)$ in place of $h_i$.

\begin{theorem}[Predictive reliability under ridge]
\label{thm:ridgeslope}
Let $E[r_{oos}(\kappa^*)]$ be the expected leave-one-out out-of-sample correlation of Definition~\ref{def:Sloo} evaluated at $\kappa=\kappa^*$, and let $E[R^2_{oos}(\kappa^*)]$ be the quantity of Theorem~\ref{thm:ridgeR2}. Under the isotropic prior of Theorem~\ref{thm:ridgepenalty}, applying the delta-method step of Definition~\ref{def:Sloo} and dropping $E[\bar r^{\,2}]$,
\[
E\big[r_{oos}(\kappa^*)\big] \;\approx\; \phi\big(E[R^2_{oos}(\kappa^*)]\big),\qquad \phi(x) := \frac{(n-1)x+1}{\sqrt{n\big[(n-2)x+2\big]}}.
\]
As in Theorem~\ref{thm:reliability}, the right-hand side is a lower bound in absolute value (Supplement~A.12).
\end{theorem}

\noindent\emph{Proof in Supplement~A.11.}

The ordering of $E[R^2_{oos}(\kappa^*)]$ and $E[r_{oos}(\kappa^*)]^2$ is the same as under ordinary least squares:
\[
E[r_{oos}(\kappa^*)]^2 - A \;=\; \frac{(1-A)^2}{n\big[(n-2)A+2\big]}, \qquad A := E[R^2_{oos}(\kappa^*)],
\]
non-negative whenever $E[R^2_{oos}(\kappa^*)] > -2/(n-2)$ --- effectively always --- and zero only at $E[R^2_{oos}(\kappa^*)]=1$. Under ridge the gap is small: for $E[R^2_{oos}(\kappa^*)] \ge 0$ it decreases in $E[R^2_{oos}(\kappa^*)]$ and is at most $1/(2n)$, reached at the break-even point. Under least squares the corresponding gap at break-even is about $\rho^2_{in,min}/2$ (Corollary~\ref{cor:halfcalib-ols}). Because ridge at $\kappa^*$ is the best linear unbiased predictor, its held-out predictions are close to calibrated, so the two statistics nearly agree; overfit least-squares predictions are not, so they can diverge widely.

\section{Simulation Study}
\label{sec:simulation}

Sections~\ref{sec:ols} and~\ref{sec:ridge} provided closed-form expressions for the expectations $E[R^2_{oos}]$ and $E[r_{oos}]$ under leave-one-out cross-validation, taken over the sampling distribution of the outcome with $X$ held fixed. These closed-form expressions are first-order delta method approximations, and, for ridge, require assumptions on the distribution of the $\beta$s. 

In the simulation study, we therefore probe three questions. The first is one of \emph{range}: whether the formulas remain accurate as $p+1$ approaches, and then exceeds, $n$ --- the region in which $D_{even}$ diverges under OLS. The second is one of \emph{assumptions}: Theorem~\ref{thm:ridgepenalty} obtains $\kappa^*$ by treating the $p$ coefficients as an isotropic draw, and Theorems~\ref{thm:ridgeeff} and~\ref{thm:ridgeR2} rely on that assumption, so we ask how much accuracy is lost when it is violated by a known amount --- separating a violation of it from collinearity in the design. The third is one of \emph{use}: whether $\kappa^*$, which a planner can compute before collecting data, is a good approximation to what cross-validation would find once outcome data is available and what the identities imply when a different estimator (lasso) would fit the data better than ridge does.

\subsection{Design of the simulation}
\label{sec:simdesign}

Each replicate draws an $n\times p$ covariate matrix with rows from $N(0,\Sigma)$, $\Sigma$ compound-symmetric with off-diagonal $\rho_c$. In the \emph{high-leverage} conditions, a random $10\%$ of rows are then multiplied by $4$, producing a heavy-tailed design in which those units are extreme on every covariate and carry most of the leverage. Each column is then standardized to sample mean zero and variance one, so the inflation changes the shape of the covariate distribution but not its variance. Coefficients follow one of three signal structures --- dense, exact-sparse with support $s$ or near-sparse --- and are then rescaled so that $\beta'\Sigma\beta = \rho^2_{in}\Tv$ exactly in every replicate, using population quantities only. Outcomes are $y_i = x_i'\beta+e_i$ with $e_i\sim N(0,(1-\rho^2_{in})\Tv)$ and $\Tv=1$. Supplement~B.1 gives the construction in full.

The rescaling ensures that $\rho^2_{in}$ is held constant (for different stated values). It also results in a test of ridge's assumptions. When $\Sigma\ne I$, drawing $\beta$ and then rescaling it tilts the coefficient distribution away from the leading directions of $\Sigma$, so the collinear conditions violate the isotropic prior of Theorem~\ref{thm:ridgepenalty} by construction. The exception is a single nonzero coefficient ($s=1$): its $\beta'\Sigma\beta$ does not depend on its position or sign, so the rescaling cannot tilt it. The collinear conditions with $s=1$ therefore separate collinearity in the design from a tilted prior (Supplement~B.1).

Overall, we examine $n\in\{30,60,250,600\}$ crossed with $(p+1)/n\in\{0.1,0.5,0.9,2.0\}$ (omitting two combinations at $n=30$ that leave too few covariates or residual degrees of freedom; Supplement~B.1), $\rho^2_{in}\in\{0.1,0.3,0.5,0.7,0.9\}$, the three signal structures, a sparsity ladder scaled with $p$ and four leverage/collinearity combinations: $1{,}920$ cells, $1{,}360$ admitting OLS and $560$ with $p+1>n$, at $R=1{,}000$ replicates each (two cells have $999$, after one failed replicate each). Averaged within cells, realized $D^*$ spans $0.31$--$0.996$ and $p_{eff}(\kappa^*)/p$ spans $0.019$--$0.990$; that movement is driven by leverage and dimension. Because each row of the grid's covariates is a linear transformation of a vector whose distribution does not involve $\Sigma$ --- Gaussian or Gaussian with a random scale in the high-leverage conditions --- $D^*$ there is invariant to $\Sigma$. Every design-stage quantity uses $n$, $p$, $\rho^2_{in}$ and the realized $D^*$; every realized quantity is the leave-one-out statistic of Definition~\ref{def:Sloo} computed from $(y,\hat y_{-i})$ alone via the exact PRESS identity~\eqref{eq:press}. We supplement these with two additional runs: all $1{,}920$ cells at $R=250$ with a $65$-point sweep of $\kappa$ and a cross-validated ridge penalty and a $144$-cell sparsity subgrid at $R=500$ carrying a cross-validated lasso (Supplement~B.1).

For each formula-versus-realized pair we form the per-replicate difference $d=\text{formula}-\text{realized}$, average it within a cell and report the median across cells of that per-cell mean, signed and in absolute value, alongside the Monte Carlo standard error (MCSE) computed on $d$. ``Realized'' means the average-of-ratios estimand of Definition~\ref{def:Sloo}. Comparisons between two \emph{realized} quantities (ridge at $\kappa^*$ against ridge at a cross-validated penalty; ridge against lasso) are formed as differences \emph{within} a replicate before averaging, so that the sampling variation the two share cancels.

\subsection{Results}

We provide a brief summary of the results here, focused on five high-level take-aways; full results by sample size, dimension ratio, penalty and signal structure are in Supplement~B. 

\textbf{1. The design formulas for $E[R^2_{oos}]$ and $E[r_{oos}]$ for both OLS and ridge track realized performance.} Figure~\ref{fig:simacc} plots formula against realized value, one point per cell. Panel~(a) spans more than four orders of magnitude in $R^2_{oos}$ --- from $-29$ to $+0.89$ in the near-saturated designs, hence the log-scale --- and the formula tracks the realized value along all of it. Table~B1 gives errors by sample size: at $(p+1)/n\le0.5$, median absolute error falls from $0.029$ to $0.001$ for the OLS $E[R^2_{oos}]$ formula between $n=30$ and $n=600$ and from $0.018$ to $0.002$ for its ridge counterpart. At $n=600$ every entry is within three times the median per-cell Monte Carlo standard error of its column, so the remaining discrepancy is of the order of Monte Carlo error at $1{,}000$ replicates.

\begin{figure}[tbp]
\centering
\includegraphics[width=0.92\textwidth]{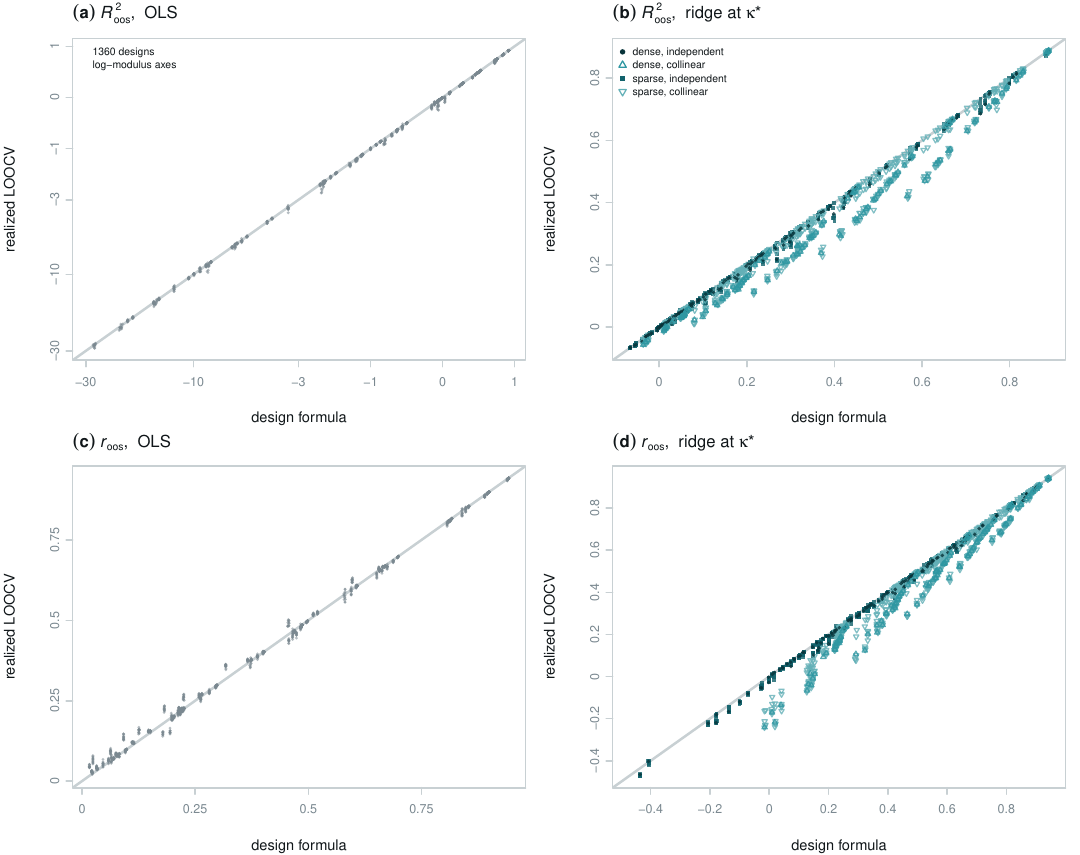}
\caption{Design formula against realized leave-one-out value, one point per cell, $1{,}920$ cells at $R=1{,}000$. Panels (a) and (c): OLS ($1{,}360$ cells; the remaining $560$ have $p+1>n$, where OLS is undefined). Panels (b) and (d): ridge evaluated at $\kappa^*$, with symbols marking the signal structure and whether the covariates are collinear. The independent designs (dark filled symbols: circles dense, squares sparse) lie on the diagonal; the departures are among the collinear designs (light open triangles), where rescaling to hold $\beta'\Sigma\beta$ fixed tilts the coefficient distribution away from isotropy. Panel (a) uses a log-scale, $\mathrm{sign}(x)\log_{10}(1+|x|)$, because realized $R^2_{oos}$ reaches $-29$ in the near-saturated cells.\label{fig:simacc}}
\end{figure}

\textbf{2. The ridge design formulas perform well, even when $p>n$; the OLS formula degrades as $p \to n$.} Binning by dimension ratio (Figure~B1), OLS $E[R^2_{oos}]$ formula error rises thirty-seven-fold, from $0.0015$ at $(p+1)/n=0.1$ to $0.056$ at $0.9$, and OLS is undefined beyond. The ridge formula varies by at most a factor of $4.4$ over the same range and continues into $p+1=2n$ at $0.012$, not far from the $0.008$ of the OLS formula at $(p+1)/n=0.5$. Because $p_{eff}(\kappa^*)$ stays well below $p$, it keeps describing realized performance where the OLS formula is undefined. The OLS degradation is anticipated by $D_{even}$ diverging as $p+1\to n$.

\textbf{3. Neither sparsity nor collinearity in the design damages the ridge design formulas; a coefficient distribution tilted away from isotropy does.} With independent predictors the prior holds in every condition, and the error is negligible whatever the signal structure: $0.003$ for dense signals and for sparse signals with $s>1$ and at most $0.001$ with $s=1$ (Table~B5). At $\rho_c=0.6$ the error rises to $0.041$ for dense signals and $0.042$ for sparse signals with $s>1$, where the rescaling tilts the coefficient distribution, but stays at $0.001$ for exact-sparse signals with $s=1$, where it cannot (\S\ref{sec:simdesign}): those designs are just as collinear, but the prior holds. The median errors are positive in every group, and the errors are positive in $57$--$100\%$ of cells, so where the formula is inaccurate it is optimistic about ridge: negligibly when the prior holds and by about $0.04$ when it does not --- small on the $R^2$ scale, but large relative to the formula's accuracy elsewhere.

\textbf{4. The design calculated $\kappa^*$ sits at the peak of the average out-of-sample curve, though cross-validation does slightly better.} Figure~\ref{fig:simkappa} sweeps the penalty from $\kappa^*/256$ to $256\kappa^*$, each cell's curve expressed as a loss from its own peak so cells of different signal strength can be averaged. The curves peak at $\kappa^*$: each cell's averaged curve peaks exactly at $\kappa^*$ in $49\%$ of cells and within one grid step in $69\%$, and the $E[R^2_{oos}]$ given up by sitting at $\kappa^*$ rather than at the summit has median $0.00001$ and maximum $0.011$.

Against a cross-validated penalty, $\kappa^*$ does slightly worse (Table~B2). Cross-validation lands close to $\kappa^*$ --- refitted to the full sample, its median penalty is between $0.99$ and $1.06$ times $\kappa^*$ at every dimension ratio --- and gains between $0.003$ and $0.025$ in $R^2_{oos}$ on average ($0.016$ over all cells), about ten Monte Carlo standard errors in each case; in no cell is $\kappa^*$ better by more than about $0.0001$. The size of the gain is set mainly by the sample size rather than by the dimension ratio: it averages $0.048$ at $n=30$, $0.022$ at $n=60$ and less than $0.01$ at $n=250$ and $600$. This is the pattern to expect if $\kappa^*$, optimal on average over coefficient vectors drawn from the prior, is outperformed by a penalty tuned to the coefficients at hand; part is optimism, as the held-out outcomes helped choose it (Supplement~B.1). The comparison therefore supports $\kappa^*$ for what it is used for here, planning a study before any outcome exists, rather than as a substitute for cross-validation once outcomes are in hand.

\begin{figure}[tbp]
\centering
\includegraphics[width=0.58\textwidth]{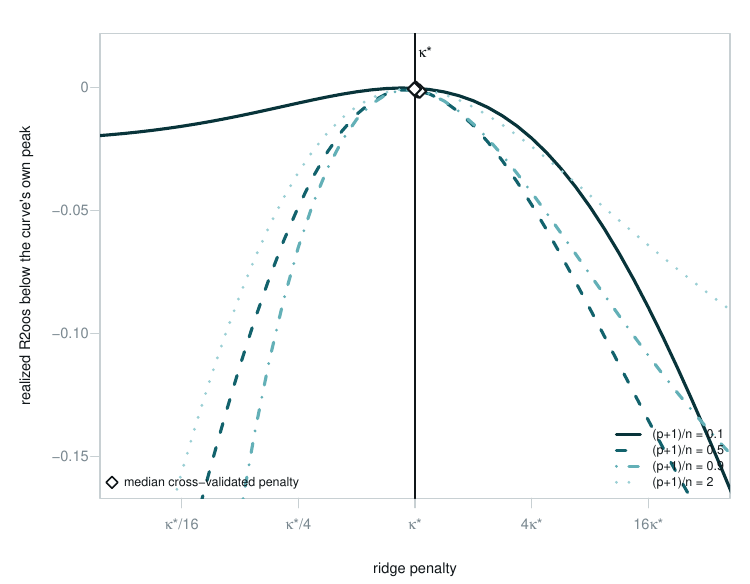}
\caption{Realized leave-one-out $R^2_{oos}$ as the ridge penalty is swept from $\kappa^*/256$ to $256\kappa^*$, averaged within dimension ratio, each cell's curve expressed as a loss from its own peak. The vertical rule is $\kappa^*$; diamonds mark the median cross-validated penalty, refitted to the full sample. The curve is asymmetric, so overshooting is more costly than undershooting; cross-validation selects penalties close to $\kappa^*$ at every dimension ratio.}\label{fig:simkappa}
\end{figure}

\textbf{5. When ridge is the wrong model, the design formula for $E[R^2_{oos}]$ becomes a lower bound.} The formulas describe what \emph{ridge} attains; if the truth is sparse, an estimator matched to sparsity (lasso) does better and the formula understates the expected values. Against ridge at $\kappa^*$ it is accurate to $+0.024$ or better under every signal structure. Against lasso the sign depends entirely on sparsity: under a dense signal the formula exceeds what lasso realizes by $0.048$; under an exactly sparse signal with $s=1$ it falls $0.162$ \emph{below}, in $91\%$ of individual replicates (Table~B3). Intermediate structures interpolate monotonically, and collinearity erodes the effect, the $s=1$ gap narrowing from $0.219$ at $\rho_c=0$ to $0.104$ at $\rho_c=0.6$. This gives the formulas a useful asymmetry: under a diffuse signal they estimate attainable performance, under a sparse one they give a lower bound on it. 

\section{An Illustration: Predicting Personality from Smartphone Behavior}
\label{sec:illustration}

We illustrate the full framework on a published prediction study. \citet{stachl2020} asked whether Big Five personality traits can be predicted from passively logged smartphone behavior. The PhoneStudy dataset contains $624$ participants, each contributing a single row of data: 1{,}821 behavioral covariates aggregated over several weeks of logging, together with self-reported personality scores. We use the data release of \citet{pargent2023}, who made these data the running example of a machine-learning tutorial, and follow their analysis choices so that our results can be compared directly with theirs. As they did, we predict the sociability facet of extraversion (the trait \citet{stachl2020} predicted most accurately) and use all $1{,}821$ covariates, replacing missing values by the covariate's median. We also follow them in excluding the four participants who did not report their gender, which leaves $n=620$. On the full set of covariates, least squares fits the data exactly in sample ($R^2=1.00$) and, under cross-validation, predicts far worse than the sample mean \citep{pargent2023}. Our formulas describe leave-one-out rather than $k$-fold cross-validation. Because least squares is not defined for the full set of covariates, where $p+1>n$, we also take nested subsets of the first $p=25$, $100$ and $600$ covariates of one random ordering. Supplement~C gives the details of the data preparation and of each step below. Every design quantity in this section can be computed with \texttt{predpower}, which includes these PhoneStudy designs as an example dataset.

\subsection{{The design efficiency $D^*$ and how to improve it}}
\label{sec:ill-dstar}

Table~\ref{tab:two} reports $D^*$ for each design. Leverage is severely uneven: most covariates are highly skewed, and $D^*$ runs from $0.18$ (at $p=600$) to $0.70$ (at $p=25$), so these designs realize between about a sixth and seven-tenths of the predictive efficiency an evenly balanced design of the same size would deliver.

Part of the problem is the imputation. Many covariates are mostly missing --- $219$ for more than $90\%$ of participants --- and median imputation turns such a covariate into a block of identical values with a few observed values far from it, the configuration that concentrates leverage on those few participants. A design analysis run before fitting would have flagged this.

Leverage concentrated by skew suggests transforming the covariates. For the $p=100$ design, $D^*=0.346$ on the raw scale rises to $0.871$ under a log transform and to $0.959$ under a rank-normal transform, and the largest leverage falls from $0.993$ to $0.508$. Supplement~C defines the transforms and shows the mechanism on the most skewed covariate.

Under ridge the design enters through $D^*(\kappa^*)$ and $p_{eff}(\kappa^*)$ instead. Because $\kappa^*$ depends on $\rho^2_{in}$, both must be evaluated at an assumed value; Table~\ref{tab:two} reports them at each design's ridge break-even point, which needs no outcome data. Shrinkage evens out leverage: $D^*(\kappa^*)$ is at least $0.95$ in every design, including the raw $p=100$ design, where it is $0.970$ against a least-squares $D^*$ of $0.346$. It also reduces the effective number of parameters sharply, to about $6$, $13$ and $24$ at $p=25$, $100$ and $600$ and to $30$ for all $1{,}821$ covariates. Ridge therefore gains little from transforming the covariates: on the rank-normal scale $D^*(\kappa^*)$ rises only to between $0.995$ and $0.997$.

\subsection{{Prospective use}}
\label{sec:ill-prospective}

None of the analyses so far has used the outcome. The remaining design parameter, $\rho^2_{in}$, must be supplied by the planner --- from substantive knowledge, from prior studies of the same outcome or by exploring a range --- just as an effect size is supplied in a power analysis. Figure~\ref{fig:prospect} shows three ways a planner might then use the formulas. Panel (a) plots the break-even $\rho^2_{in}$ against $D^*$ for this $n$ and $p$: at the raw $D^*=0.346$ least squares needs $\rho^2_{in}\ge0.361$, and after the log transform it needs $0.184$. Ridge needs only $\rho^2_{in}\ge0.027$, because its $D^*(\kappa^*)$ is close to one (Section~\ref{sec:ill-dstar}). Panel (b) traces $E[R^2_{oos}]$ across assumed values of $\rho^2_{in}$ for the log-transformed design; the gap below the dotted ceiling is what the design costs at this $n$. Panel (c) treats the study as not yet run and shows how $E[R^2_{oos}]$ under ridge climbs toward an assumed $\rho^2_{in}$ as $n$ grows. At $n=600$ and $p=100$, least squares would need $\rho^2_{in}\ge0.190$ to outperform the sample mean and ridge $0.042$. A planner would then need to assess whether such a value is plausible, given the outcome, the measures and the context.

\begin{figure}[tbp]
\centering
\includegraphics[width=\textwidth]{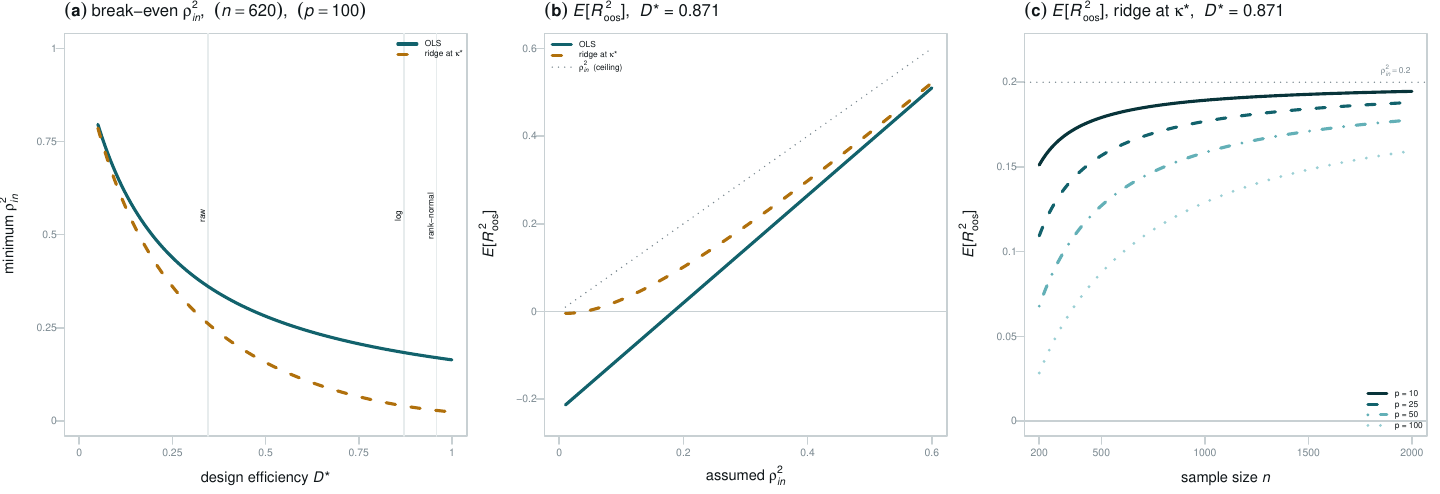}
\caption{Three prospective views, all computed from $n$, $p$, an assumed $\rho^2_{in}$ and $D^*$ alone, with $p_{eff}(\kappa^*)$ approximated as described after Corollary~\ref{cor:minridge}, so none of them uses a realized design. (a) Break-even $\rho^2_{in}$ against $D^*$ at $n=620$, $p=100$, read as $D^*(\kappa^*)$ for the ridge curve; the vertical rules mark the least-squares $D^*$ on the raw, log and rank-normal scales of Section~\ref{sec:ill-dstar} (the raw design's $D^*(\kappa^*)$ near its ridge threshold is $0.970$). (b) $E[R^2_{oos}]$ against assumed $\rho^2_{in}$ at the log-transformed $D^*=0.871$, with the ceiling $\rho^2_{in}$ dotted. (c) $E[R^2_{oos}]$ for ridge at $\kappa^*$ against $n$, at $\rho^2_{in}=0.20$ and the same log-transformed $D^*=0.871$, used for $D^*(\kappa^*)$ as well, which is conservative for ridge, for $p=10$, $25$, $50$ and $100$; the dotted line is the assumed $\rho^2_{in}$.\label{fig:prospect}}
\end{figure}

\subsection{A post-hoc analysis}
\label{sec:ill-posthoc}

Until now, we have focused on prospective use of these formulas. The PhoneStudy, however, invites a post-hoc use that we expect others with poor predictive findings may turn to: estimating $\rho^2_{in}$ from the fitted model, inserting it into the formulas and comparing the predicted $E[R^2_{oos}]$ with the realized $R^2_{oos}$. That comparison is close to circular, for the same reason post-hoc power is \citep{hoenig2001}: given $X$, the realized $R^2_{oos}$ and the in-sample $R^2_{in}$ are functions of the same residuals, by~\eqref{eq:press}, so agreement between them is not evidence that the formulas are accurate. That evidence comes from Section~\ref{sec:simulation}, where $\rho^2_{in}$ is known by construction.

A more useful comparison sets the break-even thresholds of Corollaries~\ref{cor:minR2} and~\ref{cor:minridge}, which depend only on the design, beside a confidence interval for $\rho^2_{in}$. Table~\ref{tab:two} does this for each design, on its raw scale and after the rank-normal transform. Where $p+1<n$, $\rho^2_{in}$ is estimated by the adjusted $R^2$, with an interval from the noncentral $F$ \citep{steiger1992}; the ridge threshold is found by root-finding with the design's own leverage. The realized leave-one-out values are included as description of what happened, not as a test of the formulas.

\begin{table}[tbp]
\centering
\small
\setlength{\tabcolsep}{3.5pt}
\caption{Design analysis of the PhoneStudy data ($n=620$), predicting the sociability facet of extraversion from the first $p$ covariates of one random ordering and from all $1{,}821$, on their raw scale and after a rank-normal transform. $D^*$ uses the covariates alone. The break-even columns give the $\rho^2_{in}$ each estimator needs to beat the sample mean; under ridge the penalty $\kappa^*$ depends on $\rho^2_{in}$, so $D^*(\kappa^*)$ and $p_{eff}(\kappa^*)$ are reported at that break-even value. None of these uses outcome data. $\hat\rho^2_{in}$ is the adjusted in-sample $R^2$, with a $95\%$ confidence interval from the noncentral $F$, truncated at $0$; when $p+1>n$ ($\dagger$) it is the REML estimate under the ridge prior, with a $95\%$ profile-likelihood interval (Section~\ref{sec:ill-posthoc}). The last two columns are realized leave-one-out values, ridge at the $\kappa^*$ implied by $\hat\rho^2_{in}$, which is undefined when $\hat\rho^2_{in}\le0$. A dash marks a quantity that is undefined, not one that is missing.}\label{tab:two}
\begin{tabular}{r l r r rrr r@{\,}l rr}
\toprule
& & \multicolumn{1}{c}{design} & \multicolumn{4}{c}{break-even $\rho^2_{in}$} & \multicolumn{2}{c}{} & \multicolumn{2}{c}{realized $R^2_{oos}$} \\
\cmidrule(lr){3-3}\cmidrule(lr){4-7}\cmidrule(lr){10-11}
& & & OLS & \multicolumn{3}{c}{ridge} & \multicolumn{2}{c}{} & & \\
\cmidrule(lr){5-7}
$p$ & scale & $D^*$ & needs & needs & $D^*(\kappa^*)$ & $p_{eff}(\kappa^*)$ & \multicolumn{2}{c}{$\hat\rho^2_{in}$ [95\% CI]} & OLS & ridge \\
\midrule
25 & raw & 0.696 & 0.061 & 0.014 & 0.957 & 6.4 & 0.058 & [0.023, 0.105] & $-0.009$ & $0.032$ \\
 & rank-normal & 0.984 & 0.044 & 0.013 & 0.997 & 6.0 & 0.059 & [0.023, 0.105] & $0.017$ & $0.035$ \\
100 & raw & 0.346 & 0.361 & 0.025 & 0.970 & 12.8 & 0.134 & [0.073, 0.199] & $-0.171$ & $0.079$ \\
 & rank-normal & 0.959 & 0.170 & 0.021 & 0.996 & 11.2 & 0.173 & [0.109, 0.239] & $0.004$ & $0.117$ \\
600 & raw & 0.177 & 0.994 & 0.043 & 0.983 & 24.1 & 0.048 & [0.000, 0.448] & $-155$ & $0.051$ \\
 & rank-normal & 0.722 & 0.978 & 0.034 & 0.996 & 19.3 & $\le 0$ & [0.000, 0.408] & $-42.0$ & --- \\
1{,}821 & raw & --- & --- & 0.053 & 0.984 & 30.2 & $0.221^\dagger$ & [0.133, 0.324] & --- & $0.088$ \\
 & rank-normal & --- & --- & 0.040 & 0.995 & 22.6 & $0.262^\dagger$ & [0.170, 0.368] & --- & $0.111$ \\
\bottomrule
\end{tabular}
\end{table}

For least squares the comparison is decisive at $p=100$ and above, where the threshold lies above the entire interval for $\rho^2_{in}$: at $p=100$ the design needs $\rho^2_{in}\ge0.361$, while the interval ends at $0.199$. For every value of $\rho^2_{in}$ the data support, then, this design could not be expected to beat the sample mean, and the negative realized value ($-0.171$) is what that implies. The problem is not that the predictors are uninformative: the interval puts the share of outcome variance they explain in the population between $7\%$ and $20\%$. It is that, with this $n$, $p$ and $D^*$, a least-squares fit cannot be expected to turn that information into out-of-sample accuracy. The rank-normal transform lowers the threshold at $p=100$ from $0.361$ to $0.170$, inside the interval.

For ridge the picture is different. Because shrinkage evens out leverage and reduces the effective number of parameters (Section~\ref{sec:ill-dstar}), the ridge thresholds are low, between $0.013$ and $0.053$. At $p=25$ and $p=100$ they lie below the lower confidence limits, so ridge is expected to beat the sample mean for every value of $\rho^2_{in}$ the data support; the realized values are $0.032$ and $0.079$. A cross-validated lasso gives similar values (Supplement~C).

The realized correlation tells a different story: at $p=100$ the least-squares predictions correlate $0.24$ with the held-out outcomes even though $R^2_{oos}=-0.171$. A correlation is unchanged if every prediction is stretched or shrunk toward the mean, so it shows that these predictions order participants sensibly but on far too wide a scale --- which ridge's shrinkage largely corrects ($r_{oos}=0.28$, $r^2_{oos}=0.079$, against $R^2_{oos}=0.079$). Supplement~C shows the same pattern across $p$.

For the full set of covariates, the model of \citet{pargent2023}, ridge requires $\rho^2_{in}\ge0.053$. With $p>n$ the adjusted $R^2$ is not defined, so we estimate $\rho^2_{in}$ under the same isotropic prior that defines $\kappa^*$ (Theorem~\ref{thm:ridgepenalty}). With the design held fixed, the outcome then has covariance $\nu^2 ZZ' + \Rv I$, and the two variance components can be estimated by restricted maximum likelihood (REML; \citealp{patterson1971}), as in estimates of the share of trait variance explained by genome-wide markers, where $p$ also far exceeds $n$ \citep{yang2010}. This gives $\hat\rho^2_{in}=0.221$ $[0.133, 0.324]$, well above the ridge threshold, and the realized value, $0.088$, is consistent with it. The estimate rests on the prior, and Supplement~C checks it; a bound that does not use the prior to estimate $\rho^2_{in}$ gives the same conclusion, since the $\rho^2_{in}$ of all $1{,}821$ covariates is at least that of the first $100$, whose lower confidence limit, $0.073$, already exceeds the threshold.

Finally, \citet{stachl2020} gauged their out-of-sample correlations against effect-size benchmarks, and the design formulas show why such a comparison mixes two quantities. The adjusted in-sample $R^2$, and the REML estimate when $p+1>n$, estimate $\rho^2_{in}$, a population quantity whose meaning does not change with $n$. $E[R^2_{oos}]$ is a function of $\rho^2_{in}$ and of $n$, $p$ and $D^*$ --- closer to statistical power than to a population parameter. In the PhoneStudy the $1{,}821$ covariates are estimated to explain about $22\%$ of the variance of sociability, while ridge recovers $0.088$ out of sample. Both numbers are correct, but they answer different questions. The first, by no means small, describes the relationship between smartphone behavior and sociability and is the quantity to set beside effect sizes elsewhere. The second describes what this design delivers: a larger sample would recover more of the $22\%$, as \citet{stachl2020} anticipated, and the design formulas say how much (Figure~\ref{fig:prospect}c). Whether an out-of-sample $R^2$ of this size is large enough to be useful for any decision made from these predictions remains a separate question, one the design formulas cannot answer and to which we return in Section~\ref{sec:discussion}.

\section{Discussion}
\label{sec:discussion}

In this paper, we provide formulas for the approximate expected out-of-sample accuracy $E[R^2_{oos}]$ and $E[r_{oos}]$. We show that these are determined by four quantities: the sample size $n$, the number of predictors $p$, the population $R^2$ ($\rho^2_{in}$) and the design efficiency factor $D^*$. The resulting formulas can be used to help guide the design of studies for prediction and for helping researchers better understand how these design decisions affect the performance of predictive models. 

\subsection{Relation to literature on out-of-sample statistics}

The findings of this paper are closest to those of \citet{zevallos2026}, who derive the expected out-of-sample $R^2$ under rolling- and expanding-window re-estimation. Their focus is on the equity-premium forecasting problem, and their approach is similar to ours (using the same delta-method identity). Their paper, however, does not explore the effect of uneven leverage (there, observations across time) via a design parameter, whereas we quantify this via $D^*$. It may be possible to extend our results to theirs, though we have not done so in this paper. 

For those conducting analyses with outcome data in hand, our paper also provides an approach to improve predictability through improvements to the design. In the illustration we show that by transforming the covariates in $X$, it is possible to increase the least-squares design efficiency $D^*$ and, as a result, increase both the design-stage $E[R^2_{oos}]$ and $E[r_{oos}]$ and their realized counterparts $R^2_{oos}$ and $r_{oos}$. Under ridge the gain through the design is small, because shrinkage already evens out leverage: $D^*(\kappa^*)$ was close to one before any transform. We also provide a closed-form ridge penalty $\kappa^*$. Our focus is on its use in planning: in the simulation it gave up at most $0.011$ in $R^2_{oos}$ against the best penalty in hindsight, and a cross-validated penalty did slightly better, by $0.016$ on average and mostly in small samples (Section~\ref{sec:simulation}, result~4).

\subsection{Implications for scientists}

These formulas show that, in general, $E[R^2_{oos}] \leq \rho^2_{in}$ and that the distance between the two depends upon $n$, $p$ and $D^*$. In the language of science, this means that the ability to \emph{predict} well to new units outside your sample is bounded by how well the predictors \emph{explain} the outcome in the population. The classical shrinkage formulas (Section~\ref{sec:existing}) already imply this for refit-and-replicate prediction; what the formulas here add is that, under leave-one-out, the gap depends on how the covariates are arranged, through $D^*$, and not on $n$ and $p$ alone. This suggests that prediction is best attempted where explanatory power, $\rho^2_{in}$, is already strong and well understood.

The formulas also clarify that both $R^2_{oos}$ and $r_{oos}$ are sample size dependent statistics --- they are values that are conditional on the sample size of training data. This is fundamentally different from $\rho^2_{in}$, a population quantity that does not depend on $n$. For scientists, this means that $\rho^2_{in}$ is a property of the population, which they collect data in a sample to understand. Larger samples lead to more precise inferences about $\rho_{in}^2$ but do not change the parameter itself. In comparison, the out-of-sample analogue, $E[R^2_{oos}]$ is fundamentally tied to properties of the sample. A different sample, of a different size, from the same population would result in the same $\rho^2_{in}$ but a different $E[R^2_{oos}]$. The formulas make this distinction explicit.

\subsection{Limitations}

The framework targets leave-one-out prediction specifically; the refit-and-replicate target yields a different and simpler answer (see the discussion after Theorem~\ref{thm:ols}). At the design stage, $\rho^2_{in}$ must be supplied rather than estimated, and, like the choice of parameters in a power analysis, the accuracy of the results depends upon it being selected well. 

The formulas give the expectations of $R^2_{oos}$ and $r_{oos}$, not their sampling distributions. A design that meets the break-even threshold is expected to beat the sample mean but will do so in only about half of samples; deciding how far above the threshold a design must sit to do so reliably requires the variance of these statistics as well, for which the delta-method standard error of \citet{hawinkel2024} is a natural starting point.

The break-even threshold asks only whether a design can be expected to predict better than the sample mean; it does not say how large an out-of-sample $R^2$ must be to be worth having. Fields that predict routinely have their own answers --- clinical prediction sets sample sizes so that the expected shrinkage of the fitted model is small \citep{riley2019a}, and statistical genetics judges a polygenic score against the heritability of the trait \citep{daetwyler2008,dudbridge2013} --- but the social sciences do not yet have a comparable standard. Developing one, and expressing it as a target the design formulas can be solved for, is a question for future research.

Among penalized estimators, only ridge is treated, for reasons of analytic tractability. The simulation shows the accuracy of these formulas when the underlying assumptions are not met (Section~\ref{sec:simulation}, results~3 and~5). Under a diffuse signal the ridge formula overstates what the lasso attains, by $0.048$; under a genuinely sparse one the lasso does better, by $0.162$ at $s=1$, and the formula is then a lower bound on attainable accuracy rather than an estimate of it. In the illustration the two estimators land close (Section~\ref{sec:ill-posthoc}), consistent with a diffuse signal.

The results also assume the covariates themselves are measured without error. \citet{pratiwi2024} study a closely related design factor --- the reliability of the predictors, in the classical test-theory sense of \citet{lordnovick1968} --- for out-of-sample incremental predictive validity and find that predictor unreliability attenuates predictive validity by a mechanism distinct from collinearity or leverage and that ridge's advantage over ordinary least squares is largest precisely in the small-sample, low-reliability regime this framework is meant to inform. Extending $\rho^2_{in}$ and $D^*$ to account for a known or assumed reliability of the covariates --- relevant for the PhoneStudy application, where the behavioral covariates are themselves noisy proxies rather than clean measurements --- is a natural next step and is not attempted here.

\section{Conclusion}

Design calculations for statistical power have given researchers decades of intuition about when a study is large enough to detect an effect. No corresponding intuition exists for prediction in the kinds of data and models found in the social sciences, and in its absence a disappointing out-of-sample $R^2$ is difficult to interpret: it may indicate that the phenomenon is unpredictable, that the measures were poor or that the design ($n$, $p$, $D^*$) was just not able to support prediction in the first place. This paper provides a closed-form assessment of the third, before any outcome is observed. These formulas make clear how prediction and explanation are related and how better design of a study can facilitate better predictive modeling. 

\bigskip
\begin{center}
{\large\bf ACKNOWLEDGMENTS}
\end{center}

\noindent\textbf{Use of generative AI.} Generative AI was used throughout this project as a working tool, under the author's direction. The research question, the design framework and its central quantities (the expected out-of-sample $R^2_{\mathrm{oos}}$ as the quantity to plan for, the design factor $D^*$, the break-even $\rho^2_{\min}$ and the ridge-optimal penalty $\kappa^*$), the results to be established, the design of the simulation study and the choice and framing of the illustration are the author's own. The work proceeded iteratively: the author posed each question and set each next step, and the AI tools worked through derivations, code and text, which the author checked, corrected, rejected or revised. Claude (Anthropic) was used through claude.ai in the early stages of the project (principally Claude Opus 4.8) to work through derivations and proofs with the author and to draft early notes and code. In later stages, Claude Code (Claude Opus 5 and Claude Opus 5.5) was used to continue working through derivations and proofs with the author; to write and audit R code for the simulation study, the illustration and the accompanying R package; to search and summarize literature; to check reported numbers against the code; and to draft and edit text. ChatGPT (GPT-6, OpenAI) was used to obtain comments on a draft. The author reviewed, tested and revised all output, verified every derivation and reported number and takes full responsibility for the content of the article.

\bigskip
\begin{center}
{\large\bf SUPPLEMENTARY MATERIAL}
\end{center}

\begin{description}
\item[Supplement, Section A: Proofs.] Complete proofs of every lemma, theorem and corollary stated in the article, in the order they appear, together with the setup and supporting lemmas they share. (PDF)
\item[Supplement, Section B: Simulation, further detail.] The data-generating process and study parameters in full, additional results by sample size, dimension ratio, penalty and signal structure and the supporting tables on design efficiency and the isotropic prior. (PDF)
\item[Supplement, Section C: Application, further detail.] Further detail on the PhoneStudy analysis of the sociability facet: the effect of missing values and median imputation on the design, the REML estimate of $\rho^2_{in}$ for the full set of covariates, the mechanism by which predictor skew concentrates leverage and further detail on the post-hoc analysis, including the realized correlations. (PDF)
\item[R package:] \texttt{predpower}, implementing the design formulas for least squares and ridge: $D^*$ and $D^*(\kappa^*)$ from a covariate matrix, the expected out-of-sample $R^2$ and correlation, the break-even thresholds and solvers for the sample size or number of predictors needed to reach a given $E[R^2_{oos}]$. (R source package\if0\blind{}; \url{https://github.com/marianolake/predpower}\fi)
\item[Replication code:] R code reproducing the simulation study and the application, including the design-stage calculations for $D^*$, $\kappa^*$ and the predicted design-oriented statistics. (zip archive)
\end{description}

\bibliographystyle{apalike}

\clearpage
\foreach \supppage in {1,...,27}{%
  \thispagestyle{empty}%
  \begin{tikzpicture}[remember picture,overlay]%
    \node at (current page.center) {\includegraphics[page=\supppage]{supplement.pdf}};%
  \end{tikzpicture}%
  \clearpage}

\end{document}